\documentclass[10pt, conference, letterpaper]{IEEEtran}

\IEEEoverridecommandlockouts                 

\usepackage[utf8]{inputenc}

\usepackage{amsthm, amsmath, amsfonts, amssymb}

\usepackage{enumerate}
\usepackage{graphicx}
\usepackage{mathtools}
\usepackage{thmtools}
\usepackage{thm-restate}
\usepackage{float}
\usepackage{enumerate}
\usepackage{marginnote}
\usepackage[bookmarks=false]{hyperref}
\usepackage{cleveref}
\usepackage{subfigure}
\usepackage{xcolor}

\usepackage[shortlabels]{enumitem}

\usepackage{autonum}

\usepackage{resources/custom_commands}
\usepackage{resources/coloredboxes}

\definecolor{olivedrab}{rgb}{0.42, 0.56, 0.14}

\definecolor{palatinatepurple}{rgb}{0.41, 0.16, 0.38}

\newif\ifshortversion
\shortversionfalse

\begin{document}

\title{A Model-Driven Lossless Compression Algorithm Resistant to Mismatch\\
{\footnotesize }
\thanks{This work was supported in part by the MIT Undergraduate Research Opportunities Program (UROP) and the DARPA AIQ Program. C. Hu is with the Department of Civil and Environmental Engineering and the Department of Electrical Engineering and Computer Science at the Massachusetts Institute of Technology, Cambridge, MA, USA (email: {corhu@mit.edu}). J. Tang is with the Department of Mathematics and Computer Science at the College of the Holy Cross, Worcester, MA, USA (email: {jtang@holycross.edu}).
}
}

\author{Cordelia Hu, Jennifer Tang}


\maketitle



\begin{abstract} 
Due to the fundamental connection between next-symbol prediction and compression, modern predictive models, such as large language models (LLMs), can be combined with entropy coding to achieve compression rates that surpass those of standard compression algorithms. However, this approach relies on the assumption that the predictive model produces identical output distributions at both the encoder and decoder, since even small mismatches can cause the decoding to fail. This assumption often fails with complex predictive models, particularly those based on neural networks, a phenomenon referred to as \emph{non-determinism}.
In this work, we propose a new compression algorithm based on next-token prediction that is robust to arbitrarily large, but structured, prediction mismatches. We prove the correctness of the proposed scheme under a formal mismatch certification, characterize its theoretical performance, and validate it experimentally on real datasets. Our results demonstrate reliable operation within the certified mismatch regime while achieving compression ratios that exceed those of commonly used compression methods.

\end{abstract}

\section{Introduction}

This work considers the problem of finding an algorithmic solution to non-determinism in the setting of lossless data compression. For many powerful learned predictive models, such as Large Language Models (LLMs), getting identical outputs under the same inputs is not guaranteed on different machines. This causes critical failures when predictive models are used in traditional lossless data compression pipelines.



\subsection{Model-driven Compression}

Lossless source coding is a classic problem where a discrete-time stochastic source emitting symbols is given, and the objective is to represent its output sequence using as few bits as possible while allowing exact reconstruction. 
Given an next-symbol probabilities, techniques such as Huffman coding and arithmetic coding can compress sources to nearly the entropy, essentially achieving optimal compression rate. 

From this perspective, the compression performance on real data, e.g. text, image, or videos, depends critically on how accurately the true source distribution is approximated. For example, one simple method of approximating the true source distribution for English text (used by Shannon \cite{shannon1950}) is to empirically build Markov models from sample text, where the current symbol's probability is only determined by its past $n$ symbols. 
Other commonly used compression algorithms, such as Lempel-Ziv-Welch (LZW), gzip, bzip2, and so forth, can be interpreted as implicitly building predictive distributions over future symbols as they process past symbols. 

Advances in machine learning have dramatically expanded the expressive power of learned predictive models. 
When combined with arithmetic coding \cite{pasco1976, rissanen1976, guazzo1980}, these models define a class of \emph{model-driven} compression schemes, which have demonstrated strong empirical performance in terms of compression rate \cite{knoll2014, cox2016, goyal2018deepzip, bellard2019, liu2019decmac, bellard2021, mao2022trace}.

The recent success of LLMs has further intensified interest in model-driven compression. LLMs are believed to model sources like human language very closely,
leading to LLM-driven compression systems such as LLMZip \cite{valmeekam2023llmzip} and Llama-Zip \cite{llama-zip}. LLMZip even uses LLM-driven compression to give new upper bounds on the entropy of English. 
In \cite{deletang_2024}, the authors demonstrate that compression using pre-trained models, including Llama~2 and Chinchilla, can outperform state-of-the-art general-purpose compressors on text, image and audio data. 
Additionally, model-driven compression has shown promise in lossless image compression \cite{lstm-images-2016, schiopu2018cnn, mentzer2019practical, rhee2022lc, chen2024large}, and also in more specific domains like power datasets \cite{lstm-power-data}, and neural network checkpoints \cite{kim2025-lstm-nn-checkpoints}. 
As long as there exist trained or fine-tuned models, any dataset in any domain can be compressed using these learned models in a way that can significantly narrow the gap to achieving information-theoretic limits, without needing further processing or any additional domain specific knowledge.

However, despite achieving the most competitive compression rates, LLM-driven compression has serious downsides preventing it from being readily adopted. 
First, LLM models are very large, so if the model itself needs to be transferred between the encoder and decoder, this could negate the low compression rate achieved by the entropy coder. 
Second, the need to perform a forward inference pass for each token can cause the compression to be extremely slow. 
In contrast, algorithms like gzip are essentially instantaneous even when compressing 10 MB file sizes  \cite{mittu2024finezip}. 
While these are critical issues, we do not address them in this work, as there is currently significant ongoing effort on improving LLM runtime efficiency. These efforts are not specifically targeted at compression, but model-driven compression will directly benefit from any such advances. 
Instead, in this work, we focus on a different limitation faced by LLM-driven data compression, one that is relatively overlooked yet serious. This limitation is the issue of non-determinism. 


\subsection{Non-Determinism and Prediction Mismatch}


Non-determinism refers to the phenomenon where executing the same program with identical inputs (and identical random seeds) produce different outputs on two different machines \cite{Cooper_2022, semmelrock2025}. In GPU-based inference pipelines, floating-point operations may be reordered or approximated differently in different hardware architectures, software versions, or execution environments, leading to small numerical differences that can cascade through the layers of computation and substantially alter predicted probabilities \cite{Shanmugavelu2025, Chen_2022}.
The prevalence of such effects have been documented in both empirical and theoretical studies \cite{morin2020, Eryilmaz2024, schlogl2023causesNumerical, Coakley_2022, atil2025llmdeviation}. GPU libraries 
even explicitly state that determinism is not guaranteed in many settings 
\cite{nvidia2025cublas, nvidia2025cudnn}. 

Model-driven compression requires identical model instances at the encoder and decoder ends. 
When these run on different hardware or software, non-determinism can produce slightly different conditional distributions, which often results in failures (nonsensical decompressed files). 


\subsection{Previous Work}

Recently, an algorithm was introduced called PMATIC (Probability Matched Interval Coding) \cite{pmatic1_arxiv} to address the problem of prediction mismatch in model-driven compression. PMATIC handles prediction mismatch by quantizing the space of probability distributions before applying arithmetic coding, so the encoder and decoder will quantize to matching values. PMATIC sends `helper bits' to handle the uncommon cases where the encoder and decoder are at risk of falling into different quantization bins. However, PMATIC is fundamentally based on arithmetic coding, and thus inherits its extreme sensitivity to numerical precision. Additionally, the robustness provided by PMATIC is theoretically limited to relatively rigid mismatch models and small amounts of mismatch (logit error of at most $\approx 0.02$). Lastly, PMATIC does not flexibly accommodate other statistical properties encountered in real-world learned inference pipelines.






\subsection{Main Contributions}

In this work, we develop and analyze a new algorithm for handling structured mismatch in predictions between the encoder and decoder when used for model-driven compression. Unlike PMATIC, the algorithm does not use arithmetic coding and is instead more similar to Huffman coding. While our algorithm assumes a similar error structure for the majority of token probabilities as PMATIC, it relaxes some of the restrictions PMATIC needs; our algorithm has the flexibility the adjust to different statistical assumptions, especially for the most probable tokens. We prove the theoretical correctness of our algorithm given assumptions on the mismatch structure; bound the expected encoding length; and analyze optimal parameter settings.

We also validate our algorithm with experiments on real data in order to evaluate compression performance. We show that our algorithm's compression ratio out-performs gzip, a current industry standard. Since our theoretical assumptions do not always hold with real data, we provide experiments evaluating accuracy on a typical setting with non-determinism. 

\ifshortversion
\else
This paper is organized as follows: the problem setting is presented in \Cref{sec::setting} and our algorithm is given in \Cref{sec::algorithm}. Analysis of our algorithm is given in \Cref{sec::performance} and \Cref{sec::interval}. Experimental work is shown in \Cref{sec::experiment_cordelia}.
\fi

\subsection{Notation}
\label{sec::notation}
%
%

Since file sizes are denoted in bits, for this work $\log$ is base-$2$ by default, and the natural logarithm will be denoted by $\ln$.
We denote the entropy of random variable $X$ over a finite set $S$ by $H[X] = \sum_{x \in S} -\bbP[X = x] \log(\bbP[X = x]) $. 
\ifshortversion
\else
For continuous variables $X$ with PDF $f(x)$, we express the differential entropy as $h[X] = - \int_{S} f(x) \log f(x) \,dx$.
\fi

We also define some key binary string operations. For binary strings $a = a[1] \dots a[k]$ and $b = b[1] \dots b[\ell]$:
\begin{itemize}
    \item We denote the bitwise negation of $a$ by $\overline{a}$.

    \item We denote the length of $a$ by $|a|$. For integers $1 \leq i \leq j \leq k = |a|$, we denote the substring from the $i$th to $j$th index (inclusive) by $a[i:j] := a[i] \dots a[j]$.

    \item $a \oplus b = a[1] \dots a[k] b[1] \dots b[\ell]$ denotes the concatenation of $a$ and $b$ (in order); for a concatenation of a sequence of many binary strings $a_m, a_{m+1}, \dots, a_n$, we write $\bigoplus_{i=m}^n a_i := a_m \oplus a_{m+1} \oplus \dots \oplus a_n$.
    
    \item $a \star b = \max(j : a[i] = b[i] \text{ for all } i \leq j)$ denotes the length of the longest prefix shared by $a$ and $b$.

\end{itemize}

\section{Problem Setting and Mismatch Assumptions}

\label{sec::setting}




We assume that the data has already been tokenized into a sequence of $M$ tokens $x_1, \dots x_M$, taken from a token alphabet $\mathcal{X}$.
The encoder outputs a binary string $Y$ representing the text in compressed form. In this work, $Y$ can be expressed as $\bigoplus_{i = 1}^M y_i$, where each $y_i$ is a binary string representing the encoding of token $x_i$.
The decoder decompresses $Y$ to obtain a sequence of tokens $\hat x_1, \hat x_2, \dots, \hat x_{\hat M}$, and is correct iff $\hat M = M$ and $\hat x_i = x_i$ for each $i \in \{1, \dots, M\}$. 

Predictive models take the \emph{context} of a token $x_i$, composed of the previous tokens $x_1^{i-1} = x_1, \dots x_{i-1}$, and produce a probability distribution $p(x = x_i | x_1^{i-1})$ for all $x \in \mathcal{X}$. For brevity, we leave the context $x_1^{i-1}$ implicit and denote by $p(x)$ and $p'(x)$ the probability distributions produced for the encoder and decoder, respectively.

For our compression to be possible, we must have $p(x) \approx p'(x)$. Specifically, we assume that there is some $c \geq 1$, denoting maximum discrepancy, such that
\begin{align}
\label{eq::assumption_bounded}
\left|\ln \left(p(x)\right) - \ln \left( p'(x) \right)\right| < \ln c \quad \text{for all } x \in \mathcal{X}
\end{align}
This is similar but not identical to the mismatch assumption in \cite{pmatic1_arxiv}, which assumes a bounded \emph{additive} mismatch for the prediction \emph{logits} (from which the probabilities are derived using the softmax function).

\section{Algorithm}

\label{sec::algorithm}

We first give an overview of the algorithm, which directly uses the assumption that the encoder and decoder have (multiplicatively) similar values for the probability of each token. Note that this overview leaves out important details; the complete description is given in \Cref{sec::encoder,sec::decoder}.

First, the following parameters are set in advance for both the encoder and decoder:
\begin{itemize}
    \item The amount of mismatch certified against is given by $c$.
    \item Each token $x \in \cX$ is assigned a unique fixed-length binary string $B(x) \in \{0,1\}^\ell$, where $\ell = \lceil \log \cX \rceil$. We refer to $B(x)$ as the \emph{longform} of $x$.
    \item A partition of $[0,1]$ into a set of intervals $\cR = \{ r_1, r_2, \dots, r_{|\cR|} \}$, called \emph{buckets}. We say that a bucket $r_k = (r_{k-}, r_{k+}] \in \cR$ \emph{contains} a token $x$ if $p(x) \in r_k$ (or, for the decoder, if $p'(x) \in r_k$).\footnote{Each bucket except the one containing $0$ is open below and closed above, in order to partition $[0,1]$, with $r_{(k-1)+} = r_{k-}$ for all $k$.} These $|\cR|$ buckets are represented by prefix-free binary code $A(\cdot)$.
\end{itemize}
The algorithm will work for any choice of parameters $c$, $B$, and $\cR$, though some choices of $B$ and particularly $\cR$ will result in improved performance relative to others.\footnote{For convenience, we assume in our analysis that $B(x)$ are assigned randomly; however, they could also be assigned strategically to improve average performance. 
\ifshortversion
\else
We discuss optimal selection of $\cR$ in \Cref{sec::interval}.
\fi}

Given the next-token predictions and true next token $x_i$, the encoder uses $A(\cdot)$ to communicate which bucket contains $x_i$. Then the encoder considers an expanded version (by a factor of $c^2$) of that bucket, containing all tokens that could fall into that bucket on the decoder end, and computes the minimum prefix of the longform of $x_i$ \emph{not} shared by any other token in the expanded bucket, and sends this prefix as well.
The decoder identifies and looks in its corresponding bucket (expanded by $c$) and finds the token $x_i$ matching the given prefix, which is guaranteed to be unique under the probability mismatch bound.

The overall cost of sending $x_i$ is the amount of information needed to specify the bucket, plus the length of the minimum unique prefix. Since the bucket specification will greatly narrow down the set of potential next tokens, the minimum unique prefix will likely be substantially shorter than the full token longform (even when the bucket specification length is added), thus compressing the message.

\subsection{Encoder}
\label{sec::encoder}

We examine the behavior of the encoder at a single token $x_i$ in the input $x_1, \dots, x_M$ given context $x_1^{i-1}$ (which it has already encoded into the string $\bigoplus_{j=1}^{i-1} y_j$). We denote the next-token probability (conditional on $x_1^{i-1}$, which is left implicit) as $p(x)$. The encoder will then generate some $y_i$ which represents the true token $x_i$ and append it to the encoded string. Iterating over $i$ produces the final encoded string $Y$.

\subsubsection{Encoding Subprocess}

Let $r_{k(i)} = (r_{k(i)-}, r_{k(i)+}] \in \cR$ be the bucket containing $p(x_i)$, and let
\begin{align}
    U_i := \left \{x \in \cX \setminus \{x_i\} : \frac{r_{k(i)-}}{c^2} < p(x) < c^2 \cdot r_{k(i)+} \right \} \,,
\end{align}
representing the set of other tokens, not including $x_i$, whose probabilities are `not too far' from the range represented by $r$. The encoder then computes the maximum shared prefix length between $x_i$ and $x \in U$, denoted as:
\begin{align}
\label{eq::initial_mj}
m_i = \max_{x \in U_i} (B(x) \star B(x_i))
\end{align}
Then, the encoded binary string representing $x_i$ is
\begin{align}
    \label{eq::yi_def}
    y_i = A(r_{k(i)}) \oplus B(x_i)[1:m_i+1] \oplus \overline{B(x_i)[m_i+2]}\,.
\end{align}
Intuitively, $A(r_{k(i)})$ is the string representing bucket $r_{k(i)}$ in the prefix-free code; the string $B(x_i)[1:m_i+1]$ is used to uniquely identify $x_i$ from $U_i$; and the final bit $\overline{B(x_i)[m_i+2]}$ is used to find the end of $y_i$ and thus where to start $y_{i+1}$.

\subsection{Decoder and correctness proof}
\label{sec::decoder}


The decoder decodes string $Y = \bigoplus_{i = 1}^M y_i$ by iteratively decoding substrings $\hat{y}_1, \dots, \hat{y}_{\hat{M}}$ into tokens $\hat{x}_1, \dots, \hat{x}_{\hat{M}}$. For brevity and notational simplicity, we will assume that at decoding step $i$, the decoder has correctly decoded the tokens $x_1, \dots, x_{i-1}$ and identified their corresponding sub-strings $y_1, \dots, y_{i-1}$; we will simultaneously prove that, given Assumption \eqref{eq::assumption_bounded}, this is the case via induction (i.e. by showing it holds at each successive step). Thus, when decoding $x_i$, the decoder knows the context $x_1^{i-1}$. From this, it determines its next-token predictions $p'(x)$, as well as the correct \emph{remaining} message $Y_i := \bigoplus_{j=i}^M y_j$ that it still needs to decode.


\subsubsection{Decoder Subprocess}

The decoder needs to use $p'(x)$ and $Y_i$ to determine the following: (i) the token to decode $\hat x_i$, and (ii) the length $|y_i|$, so that it can compute $Y_{i+1}$ for the next step. In order to do this, it does the following:
\begin{itemize}
    \item Finds bucket $r_{k(i)} \in \cR$ by applying prefix-free code $A(\cdot)$ to $Y_i$; since $Y_i$ begins with $y_i$, which begins with $A(r_{k(i)})$, the unique matching element of $A(\cdot)$ to $Y_i$ is $A(r_{k(i)})$. Thus (assuming all previous steps are correct) the decoder recovers the bucket $r_{k(i)}$ used by the encoder at step $i$.
    \item The decoder then removes $A(r_{k(i)})$ from $Y_i$ to obtain
    \begin{align}
    \label{eq::def_Zj}
        Z_i = B(x_i)[1:m_i + 1] \oplus \overline{B(x_i)[m_i + 2]} \oplus Y_{i+1}
        \\ \text{and} \quad \hat{U}_i = \left\{ x \in \mathcal{X}: x \in \left(r_{k(i)-}/c, c \cdot r_{k(i)+}\right) \right\}
    \end{align}
    \item Decodes $\hat x_i = \argmax_{x \in \hat{U}_i} B(x) \star Z_i$.
    \item Determines $Y_{i+1}$ by finding the bit $\overline{B(x_i)[m_i + 2]}$, which is the first bit to \emph{disagree} between $Z_i$ and $B(x_i)$ and setting $Y_{i+1}$ to be all subsequent bits from $Y_i$.
\end{itemize}
It is easy to see that all steps besides decoding $x_i$ itself succeed, provided that all prior steps were successful. Thus, we only need to show that:
\begin{proposition}
\label{prop::correctness}
Given Assumption \eqref{eq::assumption_bounded} and the correctness of all prior steps in the decoding,
\begin{align}
    x_i = \argmax_{x \in \hat{U}_i} B(x) \star Z_i \,.
\end{align}
\end{proposition}
Since $B(x_i) \star Z_i = m_i+1$, this holds if we can show that
\begin{align}
\label{eq::others_less_mj}
\max_{x \in \hat U_i \setminus x_i} B(x) \star Z_i \leq m_i\,.
\end{align}

To prove this, we use the following lemmas:
\begin{lemma}
    Given Assumption \eqref{eq::assumption_bounded}, $x_i \in \hat U_i$
\end{lemma}

\begin{proof}
    By Assumption \eqref{eq::assumption_bounded} (bounded mismatch assumption), $p'(x_i) \in (p(x_i)/c, c \cdot p(x_i)]$. Since $p(x_i) \in (r_{k(i)-}, r_{k(i)+}]$, this means $p'(x_i) \in (\frac{r_{k(i)-}}{c}, c \cdot r_{k(i)+}]$. Thus, $x_i \in \hat U_i$. 
\end{proof}

\begin{lemma}  Given Assumption \eqref{eq::assumption_bounded}, 
    $\hat U_i \setminus \{x_i\}\subseteq U_i$\,.
\end{lemma}
\begin{proof}
    Suppose for the sake of contradiction there exists some token $x \neq x_i$ such that $x \in \hat U_i, x \not \in U_i$. Then $p'(x) \in (\frac{r_{k(i)-}}{c}, c \cdot r_{k(i)+}]$, but $p(x) \not \in (\frac{r_{k(i)-}}{c^2}, c^2 \cdot r_{k(i)+}]$. This would imply that $p(x)$ and $p'(x)$ are more than a factor of $c$ away, contradicting the assumption.
\end{proof}

Thus, we know $\hat{U_i} \setminus \{x_i\} \subseteq U_i$, so by \eqref{eq::initial_mj} we get
\begin{align}
            \max_{x \in U_i} B(x) \star B(x_i)[1:m_i + 1] = m_i
\\ \implies \max_{x \in U_i} B(x) \star Z_i = m_i
\\ \implies \max_{x \in \hat{U_i} \setminus \{x_i \}} B(x) \star Z_i \leq m_i
\end{align}
so $x_i$ indeed maximizes $B(x) \star Z_i$ in $\hat{U}_i$.

Thus, we have shown that, given Assumption \eqref{eq::assumption_bounded}, all steps compute the correct values given that all previous steps are correct and hence by induction our algorithm will correctly decode the original message.

\section{Performance Analysis}

\label{sec::performance}




Since the compressed string $y_i$ representing token $x_i$ is given by \eqref{eq::yi_def}, its length is $|y_i| = |A(r_{k(i)})| + m_i + 2$. In order to optimize the compression ratio, we want to minimize 
\begin{align}
\bbE[|Y|] = \sum_{i = 1}^M \bbE[|y_i|] = \sum_{i = 1}^M \bbE[|A(r_{k(i)})|] + \bbE[m_i] + 2\,.
\end{align}
We consider a term $\bbE[|y_i|] = \bbE[|A(r_{k(i)})|] + \bbE[m_i] + 2$. First, we note that once the buckets $\cR$ are set, we can gather statistics to determine the frequency of each bucket; then, we can optimize the prefix-free code $A(\cdot)$ by using a Huffman coding. This yields
\begin{align}
    \bbE[|A(r_{k(i)})|] - H\left[\{ p(x_i) \in r_k\}_{k = 1}^{|\mathcal{R}|}\right] \in [0, 1)\,.
\end{align}

Evaluating $m_i = \text{max} _{x \in U_i} B(x) \star B(x_i)$ is trickier. Since each token's longform $B(x)$ is assigned randomly, the expression for $m_i$ should depends on $|U_i|$.

In the following analysis, we use a slightly different implementation of the $B(x)$ table, with rounds of appending one random bit to each entry on this table until all values of $B(x)$ are distinct; thus making every entry in $B(x)$ independent. 


\begin{proposition}
\label{lem::collision}
Suppose we have a finite set $U$ of $|U|$ independent, randomly selected infinite binary strings, as well as another string $b \not \in U$. Define $\eta = \max _{b' \in U} b \star b'$. Then $\bbE[\eta] = \log|U| + \Theta(1)$, and 
\begin{align}
\lim_{|U| \to \infty} \bbE[\eta] - \log |U| \leq 0.33276\,.
\end{align}
\end{proposition}

\begin {proof}

We define $k = \lceil \log |U| \rceil$, and $\omega = \frac{|U|}{2^k}$.
Notice that 
\begin{align}
\bbE[\eta] &= \sum_{i = 1} ^ \infty \bbP[\eta \geq i] = \sum _{i = 1} ^{\infty} \left[1 -  \left( 1 - 2^{-i} \right)^{|U|} \right]\,.
\end{align}

First, consider an arbitrary integer $0 \leq \mu_1 \leq k$; then since all terms of the sum are $\leq 1$, we have:
\begin{align}
\bbE[\eta] \leq (k - \mu_1 - 1) + \sum_{i = k-\mu_1}^\infty \left[1 - \left( 1 - 2^{-i} \right)^{|U|} \right]\,.
\end{align}

Let $\epsilon(x) := (1 - \frac{1}{x})^x$, which increases monotonically in $x$, and consider another arbitrary integer $\mu_2 \geq 0$.
\begin{align}
\bbE[\eta] &\leq (k - \mu_1 - 1) + \sum_{i = k-\mu_1}^\infty \left[1 -  \left( 1 - 2^{-i} \right)^{\omega 2^k} \right] \\
& = (k - \mu_1 - 1) + \sum_{i = k-\mu_1}^\infty \left[1 -  \left( \epsilon(2^i)\right)^{\omega 2^{k - i}} \right] 
\\
&\leq (k - \mu_1 - 1) + \sum _{i = k-\mu_1} ^{k+\mu_2} \left[1 -  \left( \epsilon(2^{k - \mu_1})\right)^{\omega 2^{k - i}} \right] \\
&\quad\quad + \sum _{i = k+\mu_2+1} ^{\infty} \left[1 -  \left(\epsilon(2^k) \right)^{\omega 2^{k-i}} \right]\,.
\end{align}

Notice that for all $0 < a < 1$, $1 - a^x \leq -x \ln a$ due to the function $-a^x$ being concave down. This allows us to perform a tail bound on the last term.

\begin{align}
&\sum _{i = k+\mu_2+1} ^{\infty} \left[1 -  \left(\epsilon(2^k) \right)^{\omega 2^{k-i}} \right]\\ 
&\leq \sum _{i = k+\mu_2+1} ^{\infty} \left[-\omega 2^{k-i} \ln (\epsilon(2^k)) \right]  = -\omega 2^{-\mu_2} \ln (\epsilon (2^k))
\end{align}


Notice that, by setting the parameters $\mu_1 = \mu_2 = 0$ we already show that $\bbE[\eta] \leq k + O(1)$, proving the first part of the result. However, we seek to prove the limit as well.

\begin{align}
\lim& _{|U| \to \infty} \bbE[\eta] - \log|U|= \lim _{k \to \infty} \bbE[\eta] - \log |U| \\
&\leq
\lim_{k \to \infty} \sum _{i = k-\mu_1} ^{k+\mu_2} \left[1 -  \left( \epsilon(2^{k - \mu_1})\right)^{\omega 2^{k - i}} \right] \\
& \quad \quad - (\omega 2^{-\mu_2}) \cdot\ln (\epsilon (2^k)) - \mu_1 - 1 - \log \omega \\
& = \lim_{k \to \infty} \sum _{i = k-\mu_1} ^{k+\mu_2} \left[1 -  \left( \epsilon(2^{k - \mu_1})\right)^{\omega 2^{k - i}} \right] \\
& \quad\quad + (\omega 2^{-\mu_2}) - \mu_1 - 1 - \log \omega 
\end{align}
\begin{align}
&= \lim_{k \to \infty} \sum _{i = -\mu_1} ^{\mu_2} \left[1 -  e^{\omega 2^{-i}} \right] 
+ \omega 2^{-\mu_2} - \mu_1 - 1 - \log \omega \\
&= \sum _{i = -\mu_1} ^{\mu_2} \left[1 -  e^{\omega 2^{-i}} \right] 
+ \omega 2^{-\mu_2} - \mu_1 - 1 - \log \omega \,.
\end{align}

By setting both $\mu_1$ and $\mu_2$ to sufficiently large constants (we set them to $30$), we find numerically that for all values $\omega \in (0.5, 1]$ this expression is less than $0.33276$.
\end{proof}

Notice that our code for $x_i$ after the prefix code $A(r_{k(i)})$ has length $m_i + 2$. From \Cref{lem::collision}, $\bbE[m_i] \approx 0.33276 + \log |U_i|.$ Thus, the expected length is approximately $2.33276 + \log |U_i|$. While we lose approximately $2$ bits compared to the entropy of a uniform distribution on $|U_i| + 1$ tokens, we gain the ability to tolerate mismatches.

\ifshortversion
\else

\section{Interval Selection Analysis}

\label{sec::interval}

We will now discuss the best way to select bucket probabilities. In order to do this, however, we must investigate the probability distributions produced by our predictive model.  
This way, we can choose the set of buckets $\mathcal{R}$ that minimizes the codelength given by $\bbE[|y_i|] = \bbE[A(r_{k(i)})] + \bbE [m_i] + 2$.

Given a token $x_i$ and a probability distribution, we will write this token's \emph{ranking} as $t_i$. (So if $x_i$ is the most likely token, then $t_i = 1$). 
We will use $p^*$ to represent a model's probability distribution which takes as a parameter a token's ranking $t$. 


We claim that the token distributions (specifically from the Meta Llama 3.1 model which we will be using in our experiments) approximately follow a power law distribution. In particular, the probability of a token, $p^*(t)$, should be proportional to some power of the rank of the token within the probability distribution. In other words, $p^*(t) \propto t ^ {-\alpha}$ for some $\alpha > 0$.

This is a reasonable model to assume due to Zipf's Law and other theories for word frequency, which assert that, within a given text, the frequency of a word is inversely proportional to its rank (\cite{zipf1, zipf2}). This corresponds to a power law distribution with $\alpha \approx 1$.

We assume that when models like LLMs predict words using context, this distribution is still a power law, but one with a steeper probability curve ($\alpha > 1$). The steeper curve is because we expect LLMs to provide a prediction which is more confident in a smaller set of words, meaning the prediction for LLMs should be more concentrated on higher-probability rankings.) 

Next, we will demonstrate empirical results that support this hypothesis.



\subsection {Experimental Verification and Measurement of the Power Law}
\begin{figure}
    \centering
    \includegraphics[width=0.9\linewidth]{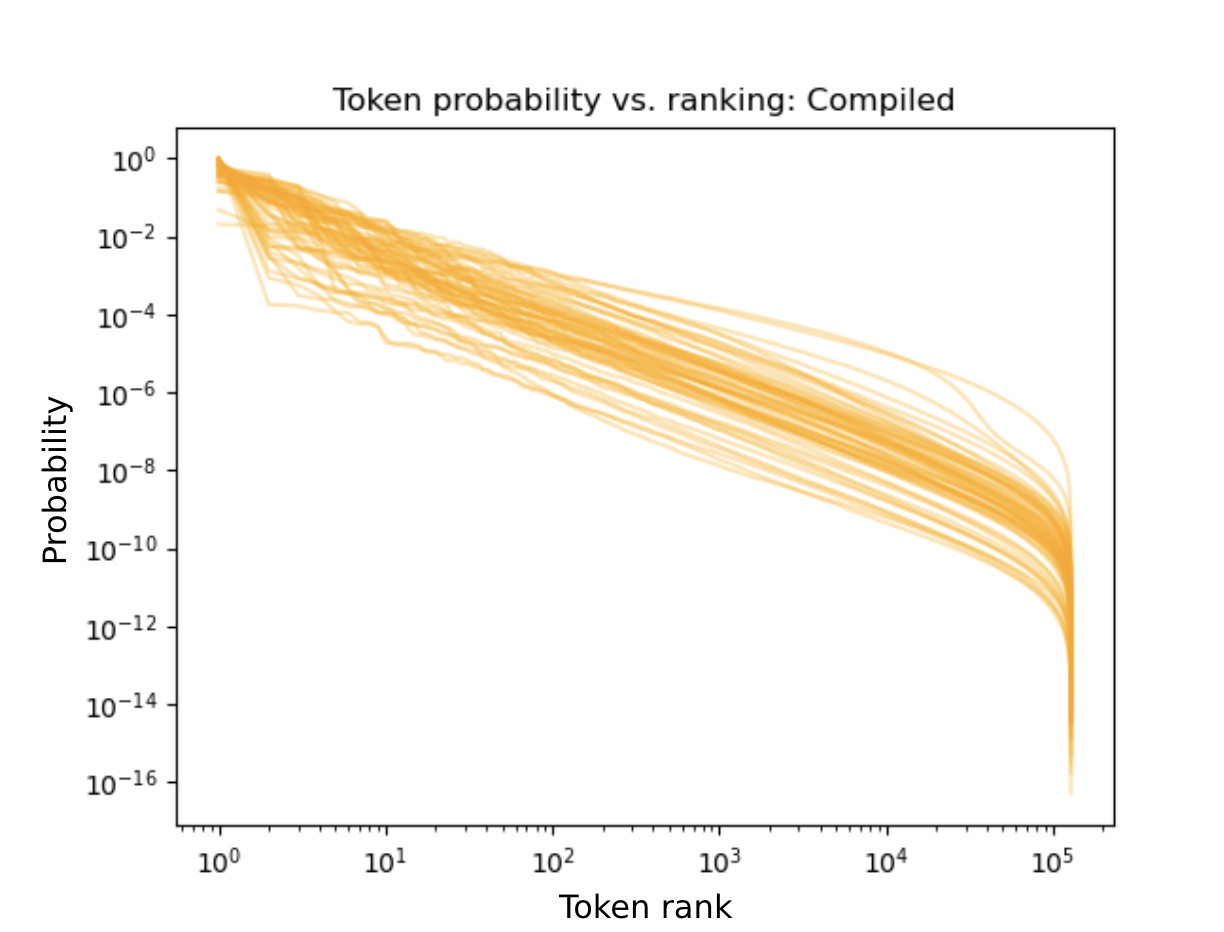}
    \caption{ Fifty probability distributions produced by Meta Llama 3.1, overlayed onto each other; snippets of Hamlet were provided as context.}
    \label{fig::loglog_prob}
\end{figure}

To determine if our power law assumption  is reasonable, we ran our predictive model (Meta Llama 3.1) using snippets of text from \emph{Hamlet} by Shakespeare. Plotting rank versus probability in a log-log plot (see \Cref{{fig::loglog_prob}}) reveals three regions of the probability distribution.

\begin{itemize}

\item The \emph{head} - this is where the top $5$ or so tokens with the highest probability (usually at least $10^{-2}$ and often close to $1$) lie. In many distributions, there is a huge dropoff between the probabilities of the most probable tokens and the rest.

\item The \emph{body} - this is where tokens with rank from $6$ to $10^5$ (inclusive) lie. They mostly follow a power law, and have probabilities usually ranging from $10^{-10}$ to $10^{-2}$ (though this depends on the actual context provided.)

\item The \emph{tail} - This is where the tokens with ranks $10^5 + 1$ and above lie. Their probabilities drop to $0$ very quickly.

\end{itemize}

While the body is the only part of the graph which obeys the power law, we decided to study it further due to its depth in token rank and high cumulative probability.

We conducted log-log regressions on the token probability distributions produced by our model using randomly selected Wikipedia aritcles as input. We only took the tokens lying in the top ten percent of token ranks (the most probable $12,825$ tokens) to specifically focus on the body. These log-log regressions allowed us to verify the strength of the power law associations and approximate the values of $\alpha$ in the bodies of each distribution.

\begin{figure}
    \centering
    \includegraphics[width=0.9\linewidth]{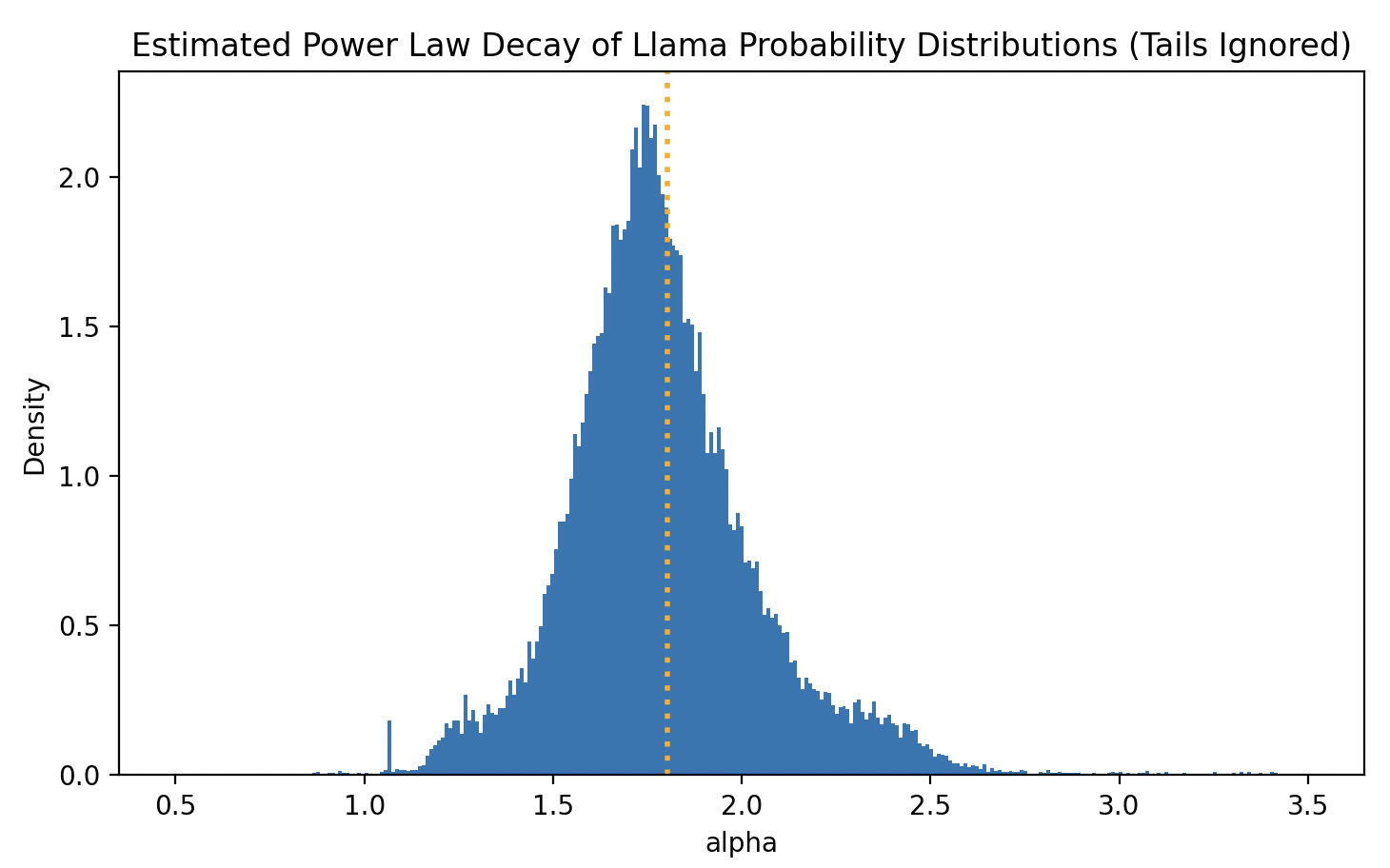}
    \caption{Regression slopes of over $100,000$ LLM probability distributions, with context from random Wikipedia articles. Slopes are estimated from linear regression after logarithmically transforming both rank and probability.}
    \label{fig:alpha}
\end{figure}

Each distribution's body strongly resembles a power law (statistical tests give mean $r^2 = 0.923$), with $\alpha$ clustering around $1.8$ (see \Cref{fig:alpha}) . Given this strong correlation, we proceeded to use our power law assumption to determine our buckets $\cR$.

\subsection{Finding Optimal Buckets}

\subsubsection{Notation and Assumptions}

We note that, though our empirical tests, we see different $\alpha$ values governing our power law approximation depending on the context given to the model. In order to come up with buckets $\mathcal{R}$, we will make some simplifying assumptions about our distribution. 

Consistent with how our encoder is defined in \Cref{sec::algorithm}
, we will assume that every probability distribution has the same shape, meaning the the same probability values are used, but the token identities can be permuted. 
For this unchanging probability shape, we want to use the power law 
\begin{align}
p^*(t) \propto t^{-\alpha} \,.
\end{align}

To make our optimization problem tractable to analytic solutions, we will additionally assume that our distribution $p^*$ is \emph{continuous}. Instead of having a set of discrete tokens, we will have a space of tokens (sorted by rank). 
This is fitting, as body token probabilities seem to match a continuous distribution quite well, and the number of tokens in the body makes the approximation quite reasonable.
Also, recall we are simplifying a collection of unknown distributions as one distribution which will already incur some amount of approximation loss; additional discrepancies between the continuous and discrete distribution are likely not significant in comparison.  

The continuous power law we will use has PDF 
\begin{align}
\label{eq::p_star_def}
p^*(t) = (\alpha - 1) \cdot t^{-\alpha}
\end{align}
defined on the space $t \in [1, \infty)$. Note that the $\alpha - 1$ coefficient comes from the fact that we want to scale $p^*$ such that $\int_1 ^{\infty} p^*(t) dt = 1$. 


We will also define a new variable $c^*$ that translates the expansion of bounds by a multiplicative factor of $c$ used in the encoding subprocess to an equivalent expansion in ranks. In particular, we want $p^*(s\cdot c^*) = \frac{1}{c} p^*(s)$. Since we want our multiplication by $c^*$ to scale probability by $\frac{1}{c}$, we set $c^* = c^{\frac{1}{\alpha}}$.

For every interval $r_k \in \mathcal{R}$ we will define $s_{k-}, s_{k+}$ such that $p^*(s_{k-}) = r_{k+}$ and $p^*(s_{k+}) = r_{k-}$. The relative ordering is flipped because the $p^*$ function is decreasing. We will also define the variable $\gamma_k^* = s_{k+}/s_{k-}$.

We will also define $\Psi$ to be the set of events $\psi_j$ corresponding to the event where $p(x_i | x_1^{i - 1}) \in r_j$. The probability of $\psi_i$ depends on $\cR$. Specifically
\begin{align}
\bbP[\psi_i] = \int_{s_{k(i)-}}^{s_{k(i)+}} p^*(t) dt \,.
\end{align}

Also, let $T$ be a random continuous variable with PDF $p^*$.

\subsubsection{Computing Buckets}

Recall that we seek to choose buckets $\cR$ that minimize the expected number of bits used to compress each token under our assumptions. From our previous results in the \Cref{sec::performance}, 
\begin{align}
\bbE[|y_i|] &= \bbE[|A_{k(i)}|] + \bbE[m_i] + 2 
\\ &\approx 2.33276 + \kappa + H[\Psi] + \bbE[\log |U_i|] \label{eq::length_yi_expression}\,.
\end{align}
Here, the letter $\kappa$ represents the Huffman encoding loss from entropy. We also use $\approx$ since there is a difference between the limit achieved in \Cref{lem::collision} and the actual value of $\bbE[|A_{k(i)}|]$ from the finite values of $|U_i|$.

From here, we will express each value in terms of the buckets in $\mathcal{R}$. We will start by simplifying the terms $H[\Psi]$ and $\bbE[\log |U_i]]$.

\begin{align}
H[\Psi] 
&= \sum_{k \in \{1, \dots, |\mathcal{R}|\}}  -\bbP[\psi_k]  \log \bbP[\psi_k]  
\\&= \sum_{k \in \{1, \dots, |\mathcal{R}|\}}  -\bbP[\psi_k]  \log \left(\int_{s_{k-}}^{s_{k+}} p^*(t) dt \right)   
\\& = \sum_{k \in \{1, \dots, |\mathcal{R}|\}}  -\bbP[\psi_k]  \log \left(s_{k-}^{-\alpha + 1} - s_{k+} ^{-\alpha - 1} \right) 
\\& = -\bbE\left[\log \left(s_{k(i)-}^{-\alpha + 1} - s_{k(i)+} ^{-\alpha + 1} \right)\right]
\\& = -\bbE\left[\log \left(s_{k(i)-}^{-\alpha + 1} - (\gamma_{k(i)}^*s_{k(i)-}) ^{-\alpha + 1} \right)\right]
\\& = (\alpha - 1) \cdot \bbE[\log(s_{k(i)-})]  \nonumber
\\&\quad \quad -\bbE\left[\log \left(1 - (\gamma_{k(i)}^*) ^{-\alpha + 1} \right)\right]
\end{align}

\begin{align}
\bbE&[\log |U_i|] \\&= \bbE\left[\log \left(\int_{(c^*)^{-2}s_{k(i)-}} ^ {(c^*)^2s_{k(i)+}} p^*(t) dt  \right) \right] 
\\&=\bbE\left[\log ((c^*)^2s_{k(i)+} - (c^*)^{-2}s_{k(i)-})\right]
\\&= 2 \log (c^*) + \bbE\left[\log (s_{k(i)+} - \frac{1}{(c^*)^4} s_{k(i)-})\right]
\\&= 2 \log (c^*) + \bbE\left[\log (\gamma_{k(i)}^* s_{k(i)-} - (c^*)^{-4} s_{k(i)-})\right]
\\&= 2 \log (c^*) + \bbE[\log(s_{k(i)-})] 
+ \bbE\left[\log (\gamma_{k(i)}^* - (c^*)^{-4}\right]\,.
\end{align}

We can now substitute in these terms into \eqref{eq::length_yi_expression}.

\begin{align}
\bbE[|y_i|] &\approx 2.33276 + \kappa + H[\Psi] + \bbE[\log |U_i|] 
\\&= \kappa + 2.33276 + 2 \log(c^*) + \alpha \bbE[\log(s_{k(i)-})]
\\&+ \bbE \left[\log \left( 
\frac{\gamma_{k(i)}^* - (c^*)^{-4}}{
1 - (\gamma_{k(i)}^*) ^{-\alpha + 1} 
}
\right) \right]\,.
\end{align}


Currently, we can isolate $\bbE[\log(s_{k(i)-})]$ as an additive term. However we want to relate this term to a quantity like $\bbE[\log t]$ which is a statistic of the distribution $p^*$ (rather than of the bucket choices). 

Notice that 

\begin{align}
\bbE[\log(t_i)] &= \sum_{k \in \{1, \dots, |\mathcal{R}|\}}  \int_{s_{k-}}^{s_{k+}} \log(t) \cdot p^* (t) dt
\\ & = \sum_{k \in \{1, \dots, |\mathcal{R}|\}} \bbP[\psi_k] \frac{\int_{s_{k-}}^{s_{k+}} \log (t) \cdot p^*(t) ds}{\bbP [\psi_k]}
\\ & = \sum_{k \in \{1, \dots, |\mathcal{R}|\}} \bbP[\psi_k] \frac{\int_{s_{k-}}^{s_{k+}} \log (t) \cdot p^*(t) ds}{
 \int_{s_{k-}}^{s^{k+}} p^*(t) ds
 }
\\& = \bbE \left[ \frac{\int_{s_{k(i)-}}^{s_{k(i)+}} \log (t) \cdot p^*(t) ds}{
 \int_{s_{k(i)-}}^{s^{k(i)+}} p^*(t) dt
 } \right] \,.
\end{align}

Notice that this inner term is the expected value of $\log _2 ( s_{k}^{\text{geom}})$, where $s_{k}^{\text{geom}}$ is the geometric mean of the terms within the interval $r_k$. In other words,

\begin{align}
 \bbE[\log(t_i)]= \bbE \left[ \log (s_{k(i)}^{\text{geom}})\right] \label{eq::expectation_gm}\,.
\end{align}

We will evaluate this geometric mean using the $\ln$ instead of $\log$ to make the derivation slightly easier.
\begin{align}
\log (s_{k(i)}^{\text{geom}}) &= \frac{\int_{s_{k(i)-}}^{s_{k(i)+}} \log (t) \cdot p^*(t) ds}{
 \int_{s_{k(i)-}}^{s^{k(i)+}} p^*(t) dt
 } \\
  \Longrightarrow \ln(s_{k(i)}^{\text{geom}})
&= \frac{\int_{s_{k(i)-}}^{s_{k(i)-}} (\alpha - 1) t^{- \alpha}\ln t  \, dt }{
\int_{s_{k(i)-}}^{s_{k(i)+}} (\alpha - 1) t^{-\alpha} dt
}\,.
\end{align}
We use integration by parts to explicitly compute the integrals.

\begin{align}
\ln(s_{k(i)}^{\text{geom}}) &= 
\frac{
-t^{1 - \alpha} \ln t\big|_{s_{k(i)-}}^{s_{k(i)+}} -
\int_{s_{k(i)-}}^{s_{k(i)+}} -t^{- \alpha}  dt
}{
-t^{1 - \alpha}\big|_{s_{k(i)-}}^{s_{k(i)+}}
}
\\&= 
\frac{
-t^{1 - \alpha} \ln t\big|_{s_{k(i)-}}^{s_{k(i)+}} +
\frac{1}{1 - \alpha} t^{1 - \alpha}\big|_{s_{k(i)-}}^{s_{k(i)+}}
}{
-t^{1 - \alpha}\big|_{s_{k(i)-}}^{s_{k(i)+}}
}
 \\&=   
\frac{
s_{k(i)-}^{1 - \alpha} \ln s_{k(i)-} - s_{k(i)+}^{1 - \alpha}\ln s_{k(i)+}
}{
s_{k(i)-}^{1 - \alpha} - s_{k(i)+}^{1 - \alpha}
} + \frac{1}{\alpha - 1} \label{eq:log_s_geom_sum}\,.
\end{align}
Focusing the first term in the sum above, we get

\begin{align}
&\frac{
s_{k(i)-}^{1 - \alpha} \ln s_{k(i)-} - s_{k(i)+}^{1 - \alpha}\ln s_{k(i)+}
}{
s_{k(i)-}^{1 - \alpha} - s_{k(i)+}^{1 - \alpha}
}
\\&=\frac{
\ln s_{k(i)-} - (\gamma_{k(i)}^*)^{1 - \alpha}\ln (\gamma_{k(i)}^* s_{k(i)-})
}{
(1 - (\gamma_{k(i)}^*)^{1 - \alpha})
}
\\&= 
\frac{
\ln s_{k(i)-} (1 - (\gamma_{k(i)}^*)^{1 - \alpha}) - (\gamma_{k(i)}^*)^{1 - \alpha} \ln \gamma_{k(i)}^*
}{
1 - (\gamma_{k(i)}^*)^{1 - \alpha}
}
\\& =  \ln s_{k(i)-} - 
\frac{
 (\gamma_{k(i)}^*)^{1 - \alpha} \ln \gamma_{k(i)}^*
}{
(1 - (\gamma_{k(i)}^*)^{1 - \alpha})
}\,.
\end{align}

Then \eqref{eq:log_s_geom_sum} is equivalent to
\begin{align}
s_{k(i)}^{\text{geom}} = s_{k(i)-} e^{\frac{1}{\alpha - 1} - \frac{
 (\gamma_{k(i)}^*)^{1 - \alpha} \ln \gamma_{k(i)}^*
}{
(1 - (\gamma_{k(i)}^*)^{1 - \alpha})
}}
\end{align}
which can be rearranged to give an expression for $s_{k(i)-}$ in terms of $s_{k(i)}^{\text{geom}}, \gamma_{k(i)}^*$, and $\alpha$, specifically
\begin{align}
\log &(s_{k(i)-}) \\&= \log  (s_{k(i)}^{\text{geom}}) - \frac{\log e}{\alpha - 1} + \log(e) \frac{
 (\gamma_{k(i)}^*)^{1 - \alpha} \ln \gamma_{k(i)}^*
}{
(1 - (\gamma_{k(i)}^*)^{1 - \alpha})
}\,.
\end{align}
%
This lets us continue our computation for
\begin{align}
&H[\Psi] + \bbE[\log |U_i|]
\\& = \bbE\left[\alpha \log (s_{k(i)-}) + 2 \log c^* + \log \left( \frac{\gamma_{k(i)}^* - \frac{1}{(c^*)^4}}{ 1 - (\gamma_{k(i)}^*)^{-\alpha + 1}} \right) \right]
\\ &= \alpha \bbE\left[\log (s^{\text{geom}}_{k(i)}) - \frac{\log e }{\alpha - 1} + \log e \cdot \frac{
 (\gamma_{k(i)}^*)^{1 - \alpha} \ln \gamma_{k(i)}^*
}{
1 - (\gamma_{k(i)}^*)^{1 - \alpha}
}\right] 
\\
&\quad + \bbE\left[2\log c^* + \log \left( \frac{\gamma_{k(i)}^* - \frac{1}{(c^*)^4}}{ 1 - (\gamma_{k(i)}^*)^{-\alpha + 1}} \right) \right]
\\& = \alpha \bbE\left[ \log (s^{\text{geom}}_{k(i)}) - \frac{1}{(\alpha - 1) \ln 2} + \frac{
 (\gamma_{k(i)}^*)^{1 - \alpha} \log \gamma_{k(i)}^* 
}{
1 - (\gamma_{k(i)}^*)^{1 - \alpha}
}\right] \\
& \quad + \bbE\left[ 2\log c^* + \log \left( \frac{\gamma_{k(i)}^* - \frac{1}{(c^*)^4}}{ 1 - (\gamma_{k(i)}^*)^{-\alpha + 1}} \right) \right] 
\\ & =  2\log c^*  + \bbE\left[\alpha \log (s_{k(i)} ^{\text{geom}} )\right] - \frac{\alpha}{(\alpha - 1) \ln 2}
\\& \quad  + \bbE\left[\alpha \cdot \frac{(\gamma_{k(i)}^*)^{1 - \alpha} \cdot \log \gamma_{k(i)}^*
}{(1 - (\gamma_{k(i)}^*)^{1 - \alpha})} + \log \left( \frac{\gamma_{k(i)}^* - \frac{1}{(c^*)^4}}{ 1 - (\gamma_{k(i)}^*)^{-\alpha + 1}} \right) \right]\,.
\end{align}

Notice now that the only parameters we have control over (in the context of choosing bins) are the $\gamma_{k}^*$ parameters. Define 
\begin{align}
&\zeta(\alpha,  c^*) = \label{eq::zeta_def} \\
&\min_{\gamma_{k(i)}} \bbE\left[\alpha \cdot \frac{
 (\gamma_{k(i)}^*)^{1 - \alpha} \log \gamma_{k(i)}
}{1 - (\gamma_{k(i)}^*)^{1 - \alpha}} + \log \left( \frac{\gamma_{k(i)}^* - \frac{1}{(c^*)^4}}{ 1 - {\gamma_{k(i)}^*}^{-\alpha + 1}} \right) \right]\,.
\end{align}

Next, we can estimate of the performance of our algorithm on assuming our idealized power law distribution.
Recall that $T$ is a random continuous variable with PDF $p^*$. Using the definition of $p^*(t)$ in \eqref{eq::p_star_def}, we have
\begin{align}
\bbE [\alpha \log t] &= \bbE\left[ \alpha \cdot \frac{1}{-\alpha} \cdot \log\left(\frac{p^*(t)}{\alpha - 1}\right)\right]  
\\ & = \bbE[\log(\alpha - 1) - \log (p^*(t))] 
\\ &
= \log (\alpha - 1) + h[T]
\end{align}


Combining this with \eqref{eq::expectation_gm} gives that \eqref{eq::length_yi_expression} can be expressed as
\begin{align}
\bbE[|y_i|] &\approx \kappa + 2.33276 +
2\log c^* + \bbE[\alpha \log (s_{k(i)} ^{\text{geom}} )]
\\& \quad \quad  - \frac{\alpha}{(\alpha - 1) \ln 2} + \zeta (\alpha,  c^*)
\\&= \kappa + 2.33276 + 2\log (c^*)  + \log (\alpha - 1) + h[T] 
\\& \quad\quad - \frac {\alpha} {(\alpha - 1) \ln 2} + \zeta (\alpha, c^*) \label{eq::expression_length_yi}
\end{align}
This gives an analytic expression that approximates the expected code length of one token. Since our assumptions is that the the output distribution of our predictive model gives the same power law distribution, this expression is meant to hold for all tokens. 

\fi

\section{Experimental Results}

\label{sec::experiment_cordelia}

\ifshortversion
\else
\subsection{Setup}
\fi




We ran our experiments on two real text datasets: randomly selected Wikipedia articles
and acts of \emph{Hamlet} by Shakespeare.
We implemented our algorithm using the Meta Llama 3.1 Q4\_K\_S LLM model to provide token probabilities. These experiments were conducted on the MIT Supercloud \cite{MITSupercloud}. To simulate non-determinism, we used two different setups for the encoding and decoding. In particular, the decoder's setup allowed it to utilize GPUs for speedup.
\ifshortversion
\else

For our tests, we used buckets \[\mathcal{R} = \{[0, 8^{-32}], (8^{-32}, 8^{-31}], \dots, (8^{-2}, 8^{-1}], (8^{-1}, 1]\}\] and a codebook $\mathcal{A}$ consisting of $\{0, 10, 110, \dots \}$. 
We calibrated $60$ versions of our compression model. Five versions of each of these sixty models are given values of $q = \frac{1}{c} \in \{0.02, 0.05, 0.1, 0.2, 0.3, 0.4, 0.5, 0.6, 0.7, 0.8, 0.9, 1\}$, and each of these five models were given a different set of longform token representation table $B$. These representations were created by a random seed-based algorithm, such that we could decompress each article with the same longform token representation table. 
\fi 


\ifshortversion
\else
\subsection{Compression Ratios}
\fi

\ifshortversion
\else
After each compression job, we converted each bitstream to bytes and stored these bytes as a text file. From there, we directly compared the compressed file sizes against the original file sizes. 
\fi 

We compute compression ratio as the original size of the file over the compressed size.
\begin{figure}
    \centering
    \includegraphics[width=0.9\linewidth]{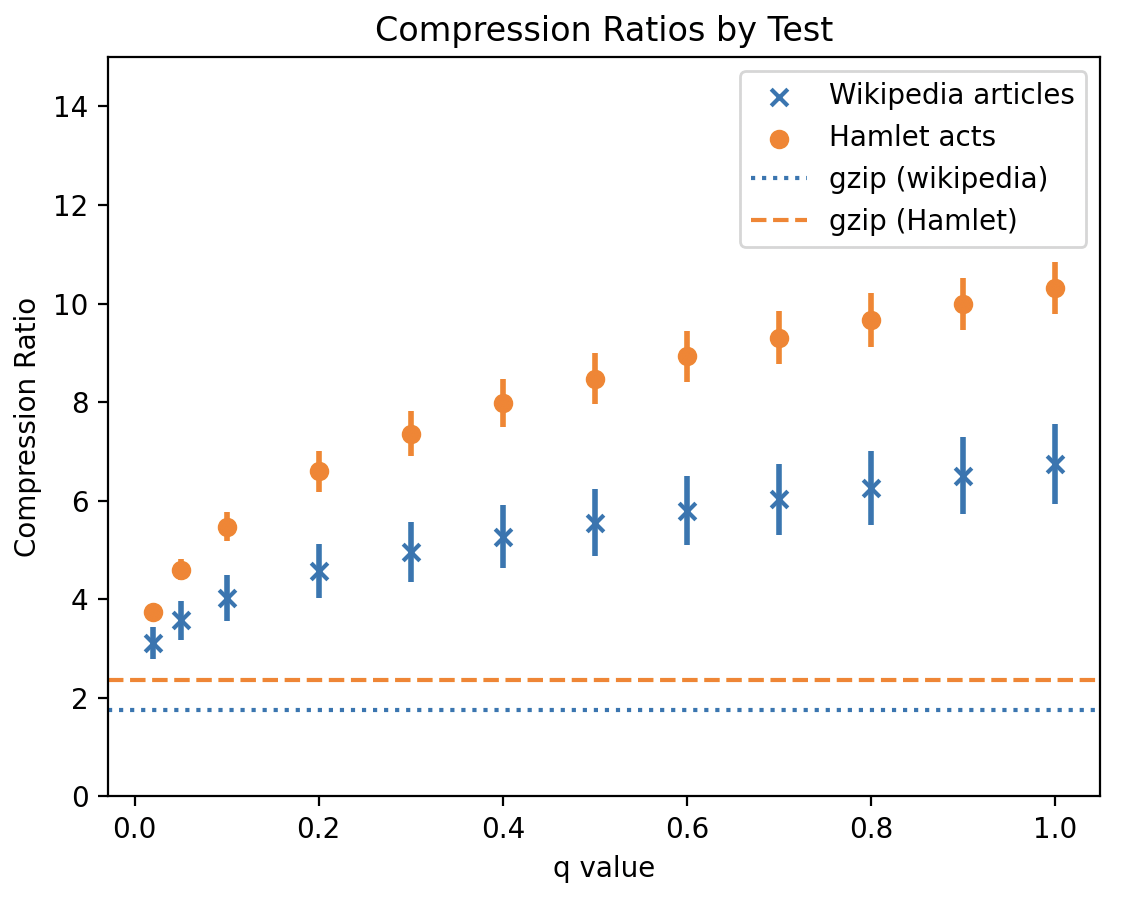}
    \caption{Compression ratios achieved in our experiements. Error bars shows the standard deviation in compression ratios. The dashed horizontal lines represent the mean compression ratios of the commonly used gzip algorithm.}
    \label{fig:placeholder}
\end{figure}
Models with a smaller $c$ naturally maintain significantly larger compression ratios than ones with larger values of $c$. However, these values of $c$ come at the cost of decreased mismatch resistance. Still, values of $c$ closer to $1$ (which are resistant to very large errors) maintain higher compression ratios than industry standards like gzip. 



\ifshortversion
\else
\subsection{Decoding Accuracy}
\fi
\begin{figure}
    \centering
    \includegraphics[width=0.9\linewidth]{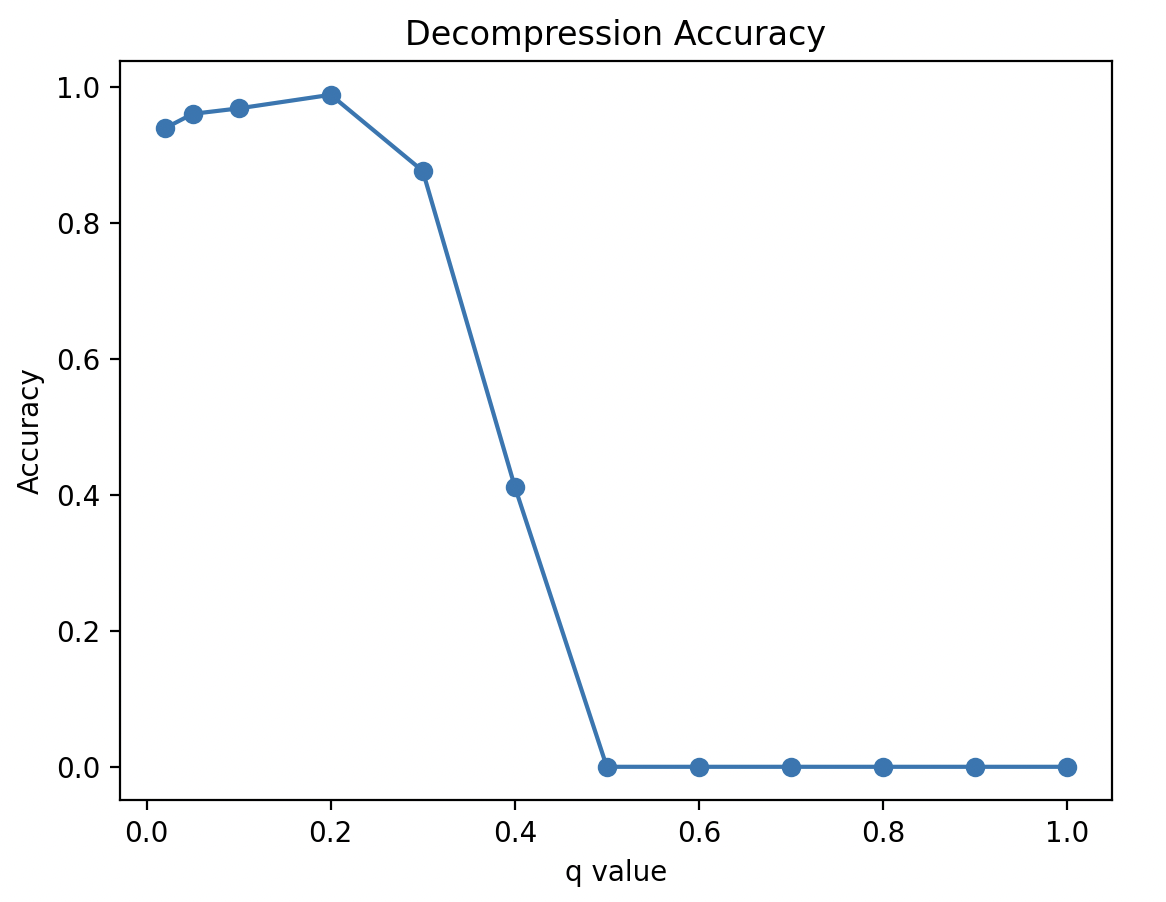}
    \caption{Decompression accuracies of Wikipedia articles, sorted by $q$ value. 
    }
    \label{fig:placeholder}
\end{figure}

Recall that while our algorithm guarantees correct decoding given Assumption \eqref{eq::assumption_bounded}, the non-determinism in our experiments does not necessarily follow this. Thus we are interested in our algorithm's performance under unknown real conditions. 

After running each test, we attempted to reconstruct the original text of each Wikipedia article with our decoding procedure. The accuracy tends to remain high until we get to $q = 0.3$ ($c = 10/3)$, where it falls off precipitously. This is not very surprising, since in preliminary testing we noticed that, comparing encoder and decoder ends, many tokens had probabilities that were mismatched by roughly a factor of $2$.

Considering that the high-accuracy $q$-values of $0.2$ and $0.3$ correspond to mean compression ratios of $4.57$ and $4.95$ respectively (far above the gzip mean compression ratio of $1.75$), we can confidently say that our model can outperform industry standards in compressing certain text files. 

Strangely, accuracy seems to increase with $q$ until $q = 0.2$, suggesting that mismatches may have additional characteristics we did not anticipate. Exploring this is a direction for further study. 




\ifshortversion
\else
\subsection{Comparison to Theoretical Expected Length}

Using values of $c$ from our experimental results, we can combine this with our analysis in \Cref{sec::interval} to estimate the expected code length per symbol under our power law assumption. 

First, we plug in our value of $\alpha = 1.804$ which we experimentally derived in \Cref{sec::interval}. We also plug in $c = \frac{10}{3}$. This value of $c$ is chosen because it is the smallest value of $c$ tested for which the compression accuracy remained high. Given $\alpha = 1.804$, we can use  $c^* = (\frac{10}{3})^{\frac{1}{\alpha}} \approx 1.95$ to estimate the value of $\gamma_k^*$ that minimizes expression for \eqref{eq::zeta_def}.  We find numerically that this is achieved at $\gamma_k^* \approx 3.748$. This value is not dependent on the index $k$, so we will set all such $\gamma^*$ to this value.

Due to our power law assumption, $\gamma^* =  3.748$ translates to a $\gamma := 3.748^{-\alpha}=0.0922$ multiplier factor on the probability bounds of each bucket.

To determine the expected length given in \eqref{eq::expression_length_yi}, we again use $\alpha$ and $c^*$ to get theoretically
%
%
\begin{align}
\bbE[|y_i|] \approx \kappa + h[T] + 5.01 \approx \kappa + 8.55
\end{align}

Experimentally, we find that mean code length per token is $7.12$. Thus, while our analysis seems to overestimate $\bbE[|y_i|]$, the difference is not massive.

\begin{figure}
    \centering
    \includegraphics[width=0.9\linewidth]{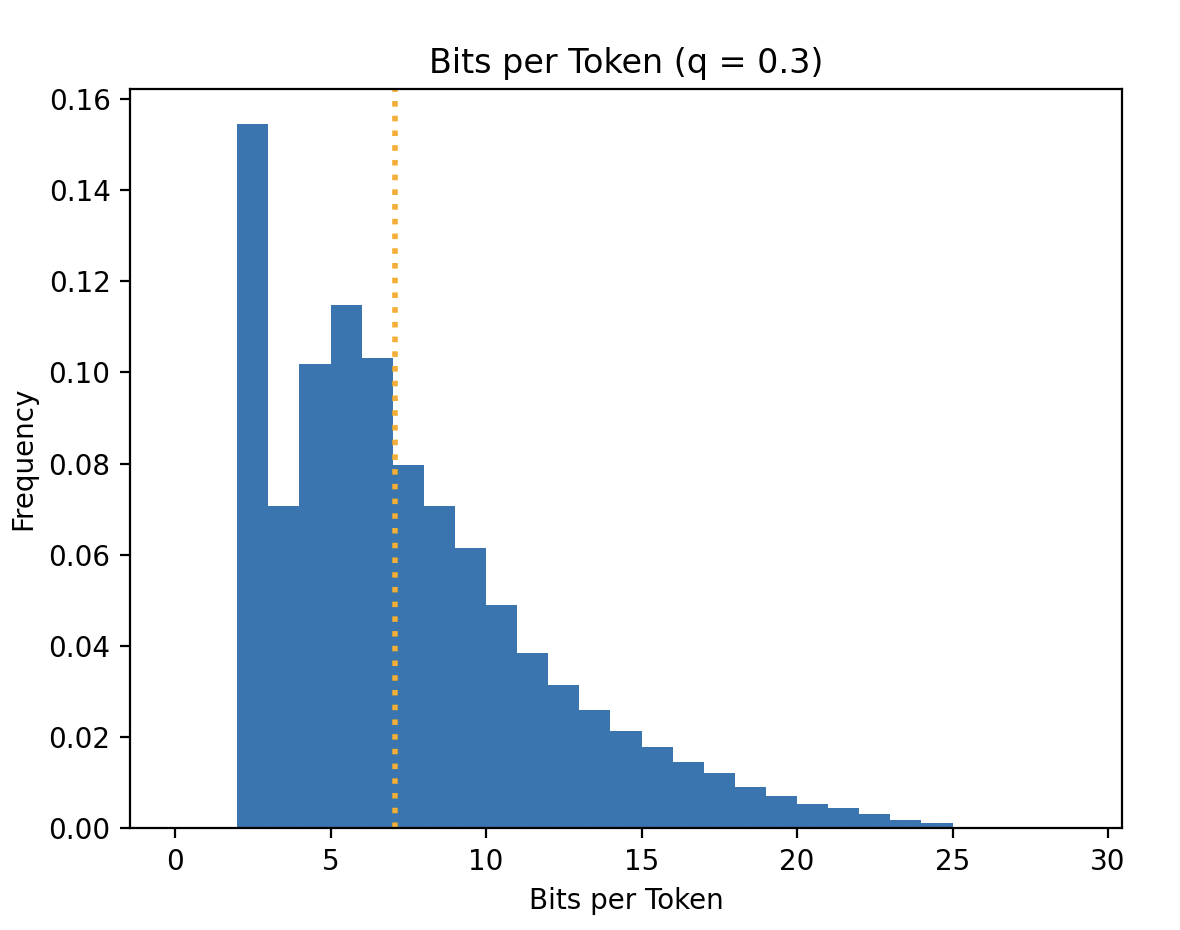}
    \caption{Histogram of the number of bits used to encode each token ($|y_i|$) in Wikipedia article tests, given $q = 0.3$. }
    \label{fig::bpt}
\end{figure}

We believe part of the disparity between the theoretical and experimental result likely comes from head terms in the distribution. Our theoretical analysis did not take their empirical distribution into consideration; we assumed they also followed a simple power law. 

Looking at the bits per token distribution confirms this disparity. In \Cref{fig::bpt} around $15$ percent of token encodings $y_i$ are two bits long. This requires both $p(x_i) > 1/8$ and $|U_i| = 0$. However, our theoretical analysis indicates that the most probable values of $|U_i|$ are $9$ and $10$, which is inconsistent with the experimental results
. Therefore, the behavior of these high probability tokens likely explains some of the performance increase from the theoretical analysis to the experimental results.

Our theoretically derived $\gamma$ value of $0.0922$ may also not actually be optimal, even given the fixed continuous distribution assumption. In our analysis, we ignored the $\kappa$ term representing the bit loss the Huffman code has when compared to theoretical entropy. 

\fi

\section{Conclusion and Future Work}

In this paper, we presented and analyzed a data compression algorithm which tolerates non-determinism. The theoretical correctness and performance of the algorithm was matched by tests on real data. On some settings, the algorithm achieved compression ratios almost three times that of gzip while retaining accuracies greater than $90 \%$. Thus our algorithm is competitive with industry-standard compression algorithms.

More investigation can be devoted to accurately pinning the behavior of the mismatch present in different scenarios. For our current test, the slight dip in accuracy for $q < 0.2$ and preliminary investigation of LLM mismatch distributions suggest that logit differences may have additional correlations.
\ifshortversion
\else

Further work can also be put into selecting bucket sizes for this algorithm. Conducting more tests that use different decaying intervals could give better intuition on how to design buckets. In addition, investigating the behavior of high-probability tokens in the distribution (which were not accounted for in this paper when analyzing bucket selection) could lead to significant improvements. Our algorithm can easily adjust the buckets sizes of the high probability tokens to account for any statistical correlations present.

\fi


\section{Acknowledgments}


The authors would like to acknowledge Ali Jadbabaie for his support of this project. 

The authors also thank Aviv Adler for his advice and comments on the manuscript. 

\newpage

\bibliographystyle{resources/IEEEbib}
\bibliography{references}

@misc{valmeekam2023llmzip,
      title={LLMZip: Lossless Text Compression using Large Language Models}, 
      author={Chandra Shekhara Kaushik Valmeekam and Krishna Narayanan and Dileep Kalathil and Jean-Francois Chamberland and Srinivas Shakkottai},
      year={2023},
      eprint={2306.04050},
      archivePrefix={arXiv},
      primaryClass={cs.IT},
      url={https://arxiv.org/abs/2306.04050}, 
}

@inproceedings{Chen_2022, series={ICSE ’22},
   title={Towards training reproducible deep learning models},
   url={http://dx.doi.org/10.1145/3510003.3510163},
   DOI={10.1145/3510003.3510163},
   booktitle={Proceedings of the 44th International Conference on Software Engineering},
   publisher={ACM},
   author={Chen, Boyuan and Wen, Mingzhi and Shi, Yong and Lin, Dayi and Rajbahadur, Gopi Krishnan and Jiang, Zhen Ming (Jack)},
   year={2022},
   month=may, pages={2202–2214},
   collection={ICSE ’22} }

@misc{semmelrock2025,
      title={Reproducibility in Machine Learning-based Research: Overview, Barriers and Drivers}, 
      author={Harald Semmelrock and Tony Ross-Hellauer and Simone Kopeinik and Dieter Theiler and Armin Haberl and Stefan Thalmann and Dominik Kowald},
      year={2025},
      eprint={2406.14325},
      archivePrefix={arXiv},
      primaryClass={cs.SE},
      url={https://arxiv.org/abs/2406.14325}, 
}

@inproceedings{Cooper_2022, 
series={CSLAW ’22},
   title={Non-Determinism and the Lawlessness of Machine Learning Code},
   url={http://dx.doi.org/10.1145/3511265.3550446},
   DOI={10.1145/3511265.3550446},
   booktitle={Proceedings of the 2022 Symposium on Computer Science and Law},
   publisher={ACM},
   author={Cooper, A. Feder and Frankle, Jonathan and De Sa, Christopher},
   year={2022},
   month=nov, pages={1–8},
   collection={CSLAW ’22} }

@misc{morin2020,
      title={Non-Determinism in TensorFlow ResNets}, 
      author={Miguel Morin and Matthew Willetts},
      year={2020},
      eprint={2001.11396},
      archivePrefix={arXiv},
      primaryClass={cs.LG},
      url={https://arxiv.org/abs/2001.11396}, 
}

@inproceedings{Coakley_2022,
   title={Examining the Effect of Implementation Factors on Deep Learning Reproducibility},
   url={http://dx.doi.org/10.1109/eScience55777.2022.00056},
   DOI={10.1109/escience55777.2022.00056},
   booktitle={2022 IEEE 18th International Conference on e-Science (e-Science)},
   publisher={IEEE},
   author={Coakley, Kevin and Kirkpatrick, Christine R. and Gundersen, Odd Erik},
   year={2022},
   month=oct, pages={397–398} }

@inproceedings{Shanmugavelu2025,
author = {Shanmugavelu, Sanjif and Taillefumier, Mathieu and Culver, Christopher and Hernandez, Oscar and Coletti, Mark and Sedova, Ada},
title = {Impacts of floating-point non-associativity on reproducibility for HPC and deep learning applications},
year = {2025},
isbn = {9798350355543},
publisher = {IEEE Press},
url = {https://doi.org/10.1109/SCW63240.2024.00028},
doi = {10.1109/SCW63240.2024.00028},
booktitle = {Proceedings of the SC '24 Workshops of the International Conference on High Performance Computing, Network, Storage, and Analysis},
pages = {170–179},
numpages = {10},
location = {Atlanta, GA, USA},
series = {SC-W '24}
}

@INPROCEEDINGS {Eryilmaz2024,
author = { Eryilmaz, Bahadir and Koras, Osman Alperen and Schlotterer, Jorg and Seifert, Christin },
booktitle = { 2024 IEEE 6th International Conference on Cognitive Machine Intelligence (CogMI) },
title = {{ Investigating the Impact of Randomness on Reproducibility in Computer Vision: A Study on Applications in Civil Engineering and Medicine }},
year = {2024},
volume = {},
ISSN = {},
pages = {265-274},
doi = {10.1109/CogMI62246.2024.00042},
url = {https://doi.ieeecomputersociety.org/10.1109/CogMI62246.2024.00042},
publisher = {IEEE Computer Society},
address = {Los Alamitos, CA, USA},
month =Oct}

@manual{nvidia2025cublas,
  title        = {NVIDIA cuBLAS Library Documentation},
  author       = {{NVIDIA Corporation}},
  year         = {2025},
  note         = {Version: CUDA cuBLAS},
  url          = {https://docs.nvidia.com/cuda/cublas/},
}

@manual{nvidia2025cudnn,
  title        = {NVIDIA cuDNN Backend API: Odds and Ends (Determinism and Reproducibility)},
  author       = {{NVIDIA Corporation}},
  year         = {2025},
  note         = {cuDNN Developer Documentation},
  url          = {https://docs.nvidia.com/deeplearning/cudnn/backend/latest/developer/misc.html},
}

@misc{knoll2014,
  author       = {Knoll, Byron},
  title        = {cmix: A Data Compression Program},
  year         = {2025},
  howpublished = {\url{https://www.byronknoll.com/cmix.html}},
  note         = {Accessed: 2025-09-23}
}

@article{cox2016,
  title={Syntactically informed text compression with recurrent neural networks},
  author={Cox, David},
  journal={arXiv preprint arXiv:1608.02893},
  year={2016}
}

@misc{deletang_2024,
      title={Language Modeling Is Compression}, 
      author={Grégoire Delétang and Anian Ruoss and Paul-Ambroise Duquenne and Elliot Catt and Tim Genewein and Christopher Mattern and Jordi Grau-Moya and Li Kevin Wenliang and Matthew Aitchison and Laurent Orseau and Marcus Hutter and Joel Veness},
      year={2024},
      eprint={2309.10668},
      archivePrefix={arXiv},
      primaryClass={cs.LG},
      url={https://arxiv.org/abs/2309.10668}, 
}

@misc{llama-zip,
  author       = {Buzanis, Alexander},
  title        = {llama-zip: An LLM-powered lossless compression tool},
  howpublished = {\url{https://github.com/AlexBuz/llama-zip}},
  year         = {2024},
  note         = {Commit on main branch; accessed 2025-05-19},
}

@phdthesis{pasco1976,
  title={Source coding algorithms for fast data compression},
  author={Pasco, Richard Clark},
  year={1976},
  school={Stanford University CA}
}

@article{rissanen1976,
  title={Generalized Kraft inequality and arithmetic coding},
  author={Rissanen, Jorma J},
  journal={IBM Journal of research and development},
  volume={20},
  number={3},
  pages={198--203},
  year={1976},
  publisher={IBM}
}

@inproceedings{schlogl2023causesNumerical,
 author = {Schl\"{o}gl, Alex and Hofer, Nora and B\"{o}hme, Rainer},
 booktitle = {Advances in Neural Information Processing Systems},
 editor = {A. Oh and T. Naumann and A. Globerson and K. Saenko and M. Hardt and S. Levine},
 pages = {56095--56107},
 publisher = {Curran Associates, Inc.},
 title = {Causes and Effects of Unanticipated Numerical Deviations in Neural Network Inference Frameworks},
 url = {https://proceedings.neurips.cc/paper_files/paper/2023/file/af076c3bdbf935b81d808e37c5ede463-Paper-Conference.pdf},
 volume = {36},
 year = {2023}
}

@misc{atil2025llmdeviation,
      title={Non-Determinism of "Deterministic" LLM Settings}, 
      author={Berk Atil and Sarp Aykent and Alexa Chittams and Lisheng Fu and Rebecca J. Passonneau and Evan Radcliffe and Guru Rajan Rajagopal and Adam Sloan and Tomasz Tudrej and Ferhan Ture and Zhe Wu and Lixinyu Xu and Breck Baldwin},
      year={2025},
      eprint={2408.04667},
      archivePrefix={arXiv},
      primaryClass={cs.CL},
      url={https://arxiv.org/abs/2408.04667}, 
}

@article{mittu2024finezip,
  title={Finezip: Pushing the limits of large language models for practical lossless text compression},
  author={Mittu, Fazal and Bu, Yihuan and Gupta, Akshat and Devireddy, Ashok and Ozdarendeli, Alp Eren and Singh, Anant and Anumanchipalli, Gopala},
  journal={arXiv preprint arXiv:2409.17141},
  year={2024}
}

@article{shannon1950,
  title={Prediction and entropy of printed English},
  author={Shannon, Claude E},
  journal={Bell system technical journal},
  volume={30},
  number={1},
  pages={50--64},
  year={1951},
  publisher={Wiley Online Library}
}

@ARTICLE{guazzo1980,
  author={Guazzo, M.},
  journal={IEEE Transactions on Information Theory}, 
  title={A general minimum-redundancy source-coding algorithm}, 
  year={1980},
  volume={26},
  number={1},
  pages={15-25},
  doi={10.1109/TIT.1980.1056143}}

@book{zipf1,
  address = {Oxford, England},
  author = {Zipf, George Kingsley},
  publisher = {Houghton Mifflin},
  shorttitle = {The psycho-biology of language},
  title = {The psycho-biology of language : an introduction to dynamic philology},
  year = 1935
}

@book{zipf2,
  author = {Zipf, George K.},
  publisher = {Addison-Wesley},
  title = {Human Behavior and the Principle of Least Effort},
  year = 1949
}

@inproceedings{MITSupercloud,
title={Interactive supercomputing on 40,000 cores for machine learning and data analysis},
author={Reuther, Albert and Kepner, Jeremy and Byun, Chansup and Samsi, Siddharth and Arcand,
William and Bestor, David and Bergeron, Bill and Gadepally, Vijay and Houle, Michael and Hubbell,
Matthew and Jones, Michael and Klein, Anna and Milechin, Lauren and Mullen, Julia and Prout,
Andrew and Rosa, Antonio and Yee, Charles and Michaleas, Peter},
booktitle={2018 IEEE High Performance extreme Computing Conference (HPEC)},
pages={1–6},
year={2018},
organization={IEEE}
}

@Article{lstm-power-data,
AUTHOR = {Ma, Zhoujun and Zhu, Hong and He, Zhuohao and Lu, Yue and Song, Fuyuan},
TITLE = {Deep Lossless Compression Algorithm Based on Arithmetic Coding for Power Data},
JOURNAL = {Sensors},
VOLUME = {22},
YEAR = {2022},
NUMBER = {14},
ARTICLE-NUMBER = {5331},
URL = {https://www.mdpi.com/1424-8220/22/14/5331},
PubMedID = {35891010},
ISSN = {1424-8220},
DOI = {10.3390/s22145331}
}

@misc{kim2025-lstm-nn-checkpoints,
      title={An Efficient Compression of Deep Neural Network Checkpoints Based on Prediction and Context Modeling}, 
      author={Yuriy Kim and Evgeny Belyaev},
      year={2025},
      eprint={2506.12000},
      archivePrefix={arXiv},
      primaryClass={cs.LG},
      url={https://arxiv.org/abs/2506.12000}, 
}

@article{lstm-images-2016,
  author       = {George Toderici and
                  Damien Vincent and
                  Nick Johnston and
                  Sung Jin Hwang and
                  David Minnen and
                  Joel Shor and
                  Michele Covell},
  title        = {Full Resolution Image Compression with Recurrent Neural Networks},
  journal      = {CoRR},
  volume       = {abs/1608.05148},
  year         = {2016},
  url          = {http://arxiv.org/abs/1608.05148},
  eprinttype    = {arXiv},
  eprint       = {1608.05148},
  bibsource    = {dblp computer science bibliography, https://dblp.org}
}

@article{chen2024large,
  title={Large Language Models for Lossless Image Compression: Next-Pixel Prediction in Language Space is All You Need},
  author={Chen, Kecheng and Zhang, Pingping and Liu, Hui and Liu, Jie and Liu, Yibing and Huang, Jiaxin and Wang, Shiqi and Yan, Hong and Li, Haoliang},
  journal={arXiv preprint arXiv:2411.12448},
  year={2024}
}

@inproceedings{schiopu2018cnn,
  title={CNN-based prediction for lossless coding of photographic images},
  author={Schiopu, Ionut and Liu, Yu and Munteanu, Adrian},
  booktitle={2018 Picture Coding Symposium (PCS)},
  pages={16--20},
  year={2018},
  organization={IEEE}
}

@inproceedings{mentzer2019practical,
  title={Practical full resolution learned lossless image compression},
  author={Mentzer, Fabian and Agustsson, Eirikur and Tschannen, Michael and Timofte, Radu and Gool, Luc Van},
  booktitle={Proceedings of the IEEE/CVF conference on computer vision and pattern recognition},
  pages={10629--10638},
  year={2019}
}

@inproceedings{rhee2022lc,
  title={LC-FDNet: Learned lossless image compression with frequency decomposition network},
  author={Rhee, Hochang and Jang, Yeong Il and Kim, Seyun and Cho, Nam Ik},
  booktitle={Proceedings of the IEEE/CVF conference on computer vision and pattern recognition},
  pages={6033--6042},
  year={2022}
}

@article{goyal2018deepzip,
  author       = {Mohit Goyal and
                  Kedar Tatwawadi and
                  Shubham Chandak and
                  Idoia Ochoa},
  title        = {DeepZip: Lossless Data Compression using Recurrent Neural Networks},
  journal      = {CoRR},
  volume       = {abs/1811.08162},
  year         = {2018},
  url          = {http://arxiv.org/abs/1811.08162},
  eprinttype    = {arXiv},
  eprint       = {1811.08162},
  bibsource    = {dblp computer science bibliography, https://dblp.org}
}

@article{bellard2019,
  title={Lossless data compression with neural networks},
  author={Bellard, Fabrice},
  journal={URL: https://bellard. org/nncp/nncp. pdf},
  year={2019}
}

@article{bellard2021,
  title={NNCP v2: Lossless data compression with transformer},
  author={Bellard, Fabrice},
  journal={Preprint at Fabrice Bellard https://bellard. org/nncp/nncp\_v2. pdf},
  year={2021}
}

@inproceedings{liu2019decmac,
  title={DecMac: A deep context model for high efficiency arithmetic coding},
  author={Liu, Qian and Xu, Yiling and Li, Zhu},
  booktitle={2019 International Conference on Artificial Intelligence in Information and Communication (ICAIIC)},
  pages={438--443},
  year={2019},
  organization={IEEE}
}

@inproceedings{mao2022trace,
  title={Trace: A fast transformer-based general-purpose lossless compressor},
  author={Mao, Yu and Cui, Yufei and Kuo, Tei-Wei and Xue, Chun Jason},
  booktitle={Proceedings of the ACM Web Conference 2022},
  pages={1829--1838},
  year={2022}
}

@misc{pmatic1_arxiv,
      title={Synchronizing Probabilities in Model-Driven Lossless Compression}, 
      author={Aviv Adler and Jennifer Tang},
      year={2026},
      eprint={2601.10678},
      archivePrefix={arXiv},
      primaryClass={cs.IT},
      url={https://arxiv.org/abs/2601.10678}, 
}



\end{document}